\documentclass[preprint,5p,11pt,twocolumn,sort&compress]{elsarticle}

\usepackage{amssymb,amsmath,colonequals,xspace,tikz}
\usetikzlibrary{calc,decorations.pathmorphing,fit,backgrounds,positioning,arrows,shapes}
\tikzstyle{knoten}=[circle,draw=black,thin,fill=white,inner sep=0pt,minimum size=4.5mm]
\tikzstyle{knotenklein}=[circle,draw=black,thin,fill=white,inner sep=0pt,minimum size=2.5mm]

\global\long\def\NP{\mathcal{NP}}

\newtheorem{theorem}{Theorem}

\newproof{proof}{Proof}
\newdefinition{definition}{Definition}

\renewcommand{\o}[1]{\overline{#1}}

\def\OO{\mathcal{O}}

\newcommand{\pt}{\ensuremath{\widetilde{p}}}

\newcommand{\floor}[1]{
	\left\lfloor #1 \right\rfloor
}

\makeatletter
\def\NAT@spacechar{~}% NEW
\makeatother

\makeatletter
\def\ps@pprintTitle{%
 \let\@oddhead\@empty
 \let\@evenhead\@empty
 \def\@oddfoot{}%
 \let\@evenfoot\@oddfoot}
\makeatother
\usepackage{mathpazo}
\usepackage[euler-digits,T1]{eulervm}
\begin{document}

\begin{frontmatter}

\title{An FPTAS for the Parametric Knapsack Problem\tnoteref{titleref}}

\author[TUKL]{Michael Holzhauser\corref{cor1}}
\ead{holzhauser@mathematik.uni-kl.de}

\author[TUKL]{Sven O. Krumke}
\ead{krumke@mathematik.uni-kl.de}

\cortext[cor1]{Corresponding author. Fax: +49 (631) 205-4737. Phone: +49 (631) 205-2511}

\address[TUKL]{\normalfont{University of Kaiserslautern, Department of Mathematics\\
  Paul-Ehrlich-Str.~14, D-67663~Kaiserslautern, Germany}}

\begin{abstract}
	\setlength{\parindent}{0pt}%
	In this paper, we investigate the parametric knapsack problem, in which the item profits are affine functions depending on a real-valued parameter. The aim is to provide a solution for all values of the parameter. It is well-known that any exact algorithm for the problem may need to output an exponential number of knapsack solutions.

	We present a fully polynomial-time approximation scheme (FPTAS) for the problem that, for any desired precision~$\varepsilon \in (0,1)$, computes $(1-\varepsilon)$-approximate solutions for all values of the parameter. This is the first FPTAS for the parametric knapsack problem that does not require the slopes and intercepts of the affine functions to be non-negative but works for arbitrary integral values. Our FPTAS outputs~$\OO(\frac{n^2}{\varepsilon})$~knapsack solutions and runs in strongly polynomial-time $\OO(\frac{n^4}{\varepsilon^2})$. Even for the special case of positive input data, this is the first FPTAS with a strongly polynomial running time. We also show that this time bound can be further improved to $\OO(\frac{n^2}{\varepsilon} \cdot A(n,\varepsilon))$, where $A(n,\varepsilon)$ denotes the running time of any FPTAS for the traditional (non-parametric) knapsack problem.
\end{abstract}


\end{frontmatter}

\section{Introduction}
\label{sec:Intro}

The knapsack problem is one of the most fundamental combinatorial optimization problems: Given a set of $n$ items with weights and profits and a knapsack capacity, the task is to choose a subset of the items with a maximum profit such that the weight of these items does not exceed the knapsack capacity. The problem is known to be weakly $\NP$-hard and solvable in pseudo-polynomial time. Moreover, several constant factor approximation algorithms and approximation schemes have been developed for the problem \citep{IbarraKimKnapsack,LawlerCombinatorialOptimization,MagazineOguzKnapsack,KellererPferschyFPTAS,KellererPferschyFPTAS2} (cf. \citep{Knapsack} for an overview).

In this paper, we investigate a generalization of the problem in which the profits are no longer constant but affine functions depending on a parameter~$\lambda \in \mathbb{R}$. More precisely, for a knapsack with \emph{capacity}~$W$ and for each \emph{item}~$i$ in the \emph{item set}~$\{1,\ldots,n\}$ with \emph{weight}~$w_i \in \mathbb{N}_{> 0}$, the \emph{profit}~$p_i$ is now of the form $p_i(\lambda) \colonequals a_i + \lambda \cdot b_i$ with $a_i,b_i \in \mathbb{Z}$. The resulting optimization problem can be stated as follows:
\begin{align*}
	p^*(\lambda) = \max\ & \sum_{i=1}^n (a_i + \lambda \cdot b_i) \cdot x_i \\
	& \sum_{i=1}^n w_i \cdot x_i \leq W \\
	& x_i \in \{0,1\} \quad \forall i \in \{1,\ldots,n\}
\end{align*}
The aim of this \emph{parametric knapsack problem} is to return a partition of the real line into intervals~$(-\infty,\lambda_1], [\lambda_1,\lambda_2], \ldots, [\lambda_{k-1},\lambda_k], [\lambda_k,+\infty)$ together with a solution~$x^*$ for each interval such that this solution is optimal for all values of $\lambda$ in the interval. The function mapping each $\lambda \in \mathbb{R}$ to the profit of the corresponding optimal solution is called the \emph{optimal profit function} and will be denoted by $p^*(\lambda)$ in the following. It is easy to see that $p^*$ is continuous, convex, and piecewise linear with breakpoints~$\lambda_1,\ldots,\lambda_k$ \citep{Knapsack}.

Clearly, since the parametric knapsack problem is a generalization of the traditional (non-parametric) knapsack problem, it is at least as hard to solve as the knapsack problem. In fact, it was shown that, even in the case of integral input data, the minimum number of breakpoints of the optimal profit function can be exponentially large, so any exact algorithm may need to return an exponential number of knapsack solutions \citep{CarstensenParametricProblems}. In this paper, we are interested in a \emph{fully polynomial time approximation scheme} for the parametric knapsack problem. We will show that, for any desired precision~$\varepsilon \in (0,1)$, a polynomial number of intervals suffices in order to be able to provide a $(1-\varepsilon)$-approximate solution for each $\lambda \in \mathbb{R}$.

Without loss of generality, we may assume that $w_i \leq W$ for each $i \in \{1,\ldots,n\}$ since otherwise we are not able to chose item~$i$ at all. However, note that we do not set any further restrictions on the profits, i.e., the parameters~$a_i$ and $b_i$. In particular, each profit may become negative for some specific value of $\lambda$. It is even possible that there a no profitable items at all for some values of $\lambda$.

\subsection{Previous work}

A large number of publications investigated parametric versions of well-known problems. This includes the parametric shortest path problem \citep{KarpOrlinParametricShortestPath,YoungTarjanOrlinParametricShortestPath,CarstensenParametricProblemsDisseration,MulmuleyShortestPathLowerBound}, the parametric minimum spanning tree problem \citep{FernandezParametricMinimumSpanningTree,AgarwalParametricMinimumSpanningTree}, the parametric maximum flow problem \citep{GalloParametricMaxFlow,McCormickParametricMaxFlow,ScutellaParametricMaxFlow}, and the parametric minimum cost flow problem \citep{CarstensenParametricProblems} (cf. \citep{SpundParametricKnapsack} for an overview). The parametric knapsack problem considered here was first investigated by \citet{CarstensenParametricProblems}. She showed that the number of breakpoints of the optimal profit function can become exponentially large in general. If the parameters are integral, the number of breakpoints can still attain a pseudo-polynomial size. The first specialized exact algorithm for the problem was presented by \citet{ChaimeParametricKnapsack}, who showed that the problem can be solved in $\OO(knW)$~time, where $k$ denotes the number of breakpoints of the optimal profit function~$p^*$.

The first approximation scheme for the problem was recently published by \citet{SpundParametricKnapsack}. The authors presented a generalization of the standard polynomial-time approximation scheme for the knapsack problem, resulting in a PTAS for the problem with a running time in $\OO(\frac{1}{\varepsilon^2} \cdot n^{\frac{1}{\varepsilon}+2})$. In the special case of positive values of $\lambda$ as well as non-negative values of $a_i$ and $b_i$ for each $i \in \{1,\ldots,n\}$, the authors show that an algorithm of \citet{ErlebachMultiObjectiveKnapsack} for the bicriteria knapsack problem can be used to obtain an FPTAS for the parametric knapsack problem running in $\OO(\frac{n^3}{\varepsilon^2} \cdot \log^2 \mathrm{UB}_{max})$~time, where $\mathrm{UB}_{max}$ denotes an upper bound on the maximum possible profit with respect to both of the profit functions $\sum_{i=1}^n a_i \cdot x_i$ and $\sum_{i=1}^n b_i \cdot x_i$.

\subsection{Our contribution}

We present the first FPTAS for the parametric knapsack problem without the restriction to non-negative input data. In particular, we show that we only need a total number of $\OO(\frac{n^2}{\varepsilon})$~intervals to approximate the problem (which, itself, may need an exponential number of intervals as described above) and that we can compute an approximate solution for each interval in $\OO(\frac{n^2}{\varepsilon})$~time, yielding an FPTAS with a strongly polynomial running time of $\OO(\frac{n^4}{\varepsilon^2})$. Our algorithm is the first FPTAS for the problem with a strongly polynomial running time, being superior to the PTAS of \citet{SpundParametricKnapsack} for $\varepsilon \leq 0.5$ and, in the special case of positive input data, superior to their FPTAS for large input values. In a second step, we improve this result to a running time of $\OO(\frac{n^2}{\varepsilon} \cdot A(n,\varepsilon))$, where $A(n,\varepsilon)$ denotes the running time of any FPTAS for the traditional knapsack problem. Using the FPTAS of \citet{KellererPferschyFPTAS,KellererPferschyFPTAS2}, this yields a time bound of $\OO\left( \frac{n^3}{\varepsilon} \log \frac{1}{\varepsilon} + \frac{n^2}{\varepsilon^4} \log^2 \frac{1}{\varepsilon} \right)$ for the parametric knapsack problem.

\subsection{Organization}

The results of this paper are divided into three main parts. In Section~\ref{sec:TwoApprox}, we show how we can generalize the well-known greedy-like $\frac{1}{2}$-approximation algorithm for the traditional knapsack problem to the parametric setting and how the resulting profit function can be ``smoothened'' such that it becomes convex and continuous without losing the approximation guarantee. This will be the key ingredient for the parametric FPTAS, which will be presented in Section~\ref{sec:FPTAS}. We will first recapitulate a basic FPTAS for the traditional knapsack problem in Section~\ref{sec:FPTAS:Traditional} and then extend it to the parametric case in Section~\ref{sec:FPTAS:Parametric}, presenting a first time bound for the resulting FPTAS. In Section~\ref{sec:FPTAS:Analysis}, as a main result of the paper, we present an improved analysis yielding the claimed running time of the FPTAS. Finally, in Section~\ref{sec:FPTAS2}, we show that is suffices to solve the corresponding subproblems only approximately so that we can incorporate traditional FPTASs, which improves the running time of our algorithm to the claimed one.

\section{Obtaining a parametric $\frac{1}{2}$-approx\-imation}
\label{sec:TwoApprox}

The parametric FPTAS will rely on a convex and continuous $\frac{1}{2}$-approximation of the optimal profit function~$p^*(\lambda)$, i.e., a parametric $\frac{1}{2}$-approximation algorithm for the parametric knapsack problem. We will therefore present such an algorithm in this section and describe how we can guarantee these properties of the function.

\subsection{Traditional $\frac{1}{2}$-approximation algorithm}
\label{sec:TwoApprox:Traditional}

The basic (non-parametric) $\frac{1}{2}$-approximation algorithm proceeds as follows: In a first step, the algorithm sorts the items in decreasing order of their ratios~$\frac{p_i}{w_i}$, which can be done in $\OO(n \log n)$~time. It then packs the items in this ordering until the next item~$k$ with $k \geq 2$ would violate the knapsack capacity (or until there are no items left), yielding a feasible solution~$x'$. The algorithm either returns $x'$ or, if better, the solution containing only an item with the largest profit~$p^{(\max)}$. If $x^*$ denotes an optimal solution to the given knapsack instance, $x^A$ the solution returned by the above algorithm, and $x^{LP}$ a solution to the LP-relaxation of the problem, we get that
\begin{align*}
	p^A \colonequals p(x^A) &\colonequals \sum_{i=1}^n p_i \cdot x^A_i = \max\left\{ p^{(\max)}, \sum_{i=1}^{k-1} p_i \right\} \\
	&\geq \frac{1}{2} \cdot \sum_{i=1}^k p_i \geq \frac{1}{2} \cdot p(x^{\mathrm{LP}}) \geq \frac{1}{2} \cdot p(x^*),
\end{align*}
so $x^A$ is a $\frac{1}{2}$-approximation. We refer to \citep{Knapsack} for further details on this standard algorithm.

\subsection{Parametric $\frac{1}{2}$-approximation algorithm}
\label{sec:TwoApprox:Parametric}

In the parametric knapsack problem, the profits are affine functions of the form~$p_i(\lambda) = a_i + \lambda \cdot b_i$ such that the optimal profit~$p^*$ changes with $\lambda$. However, note that the solution~$x'$ of the traditional $\frac{1}{2}$-approximation algorithm only depends on the ordering of the items and, thus, remains constant as long as this ordering does not change. Moreover, two items can only change their relative ordering if their profit functions intersect, yielding $\OO(n^2)$ intervals~$I'_j$, within which the ordering of the items remains unchanged. For all values of $\lambda$ in such an interval~$I'_j$, the algorithm computes the same solution~$x'$, which has a profit of the form~$p^{(j)}(\lambda) \colonequals \alpha^{(j)} + \lambda \cdot \beta^{(j)}$. For each $\lambda \in I'_j$, the $\frac{1}{2}$-approximation algorithm either returns $x'$ with profit $p^{(j)}(\lambda)$ or the most valuable item only. One possibility to obtain a parametric $\frac{1}{2}$-approximation algorithm would be to consider each interval~$I'_j$ separately and to divide it into subintervals, depending on whether $p^{(j)}(\lambda)$ or $p^{(\max)}(\lambda)$ is larger, where $p^{(\max)}$ denotes the profit of the most valuable item (which, again, now depends on $\lambda$). However, the resulting piecewise linear function~$p^A(\lambda)$ is not necessarily continuous or convex, which will be required later, though (see Figure~\ref{fig:TwoApprox}).

Instead, we ignore the intervals~$I'_j$ and only consider the above affine functions~$p^{(j)}(\lambda) = \alpha^{(j)} + \lambda \cdot \beta^{(j)}$. Let $S$ denote the set of all such functions together with the function~$p^{(0)}(\lambda) \colonequals 0$ and each profit function~$p_i(\lambda)$. By computing the upper envelope of the~$\OO(n^2)$~functions in $S$, we obtain a function~$\varphi$, which is based on feasible solutions whose profit is not smaller than $p^A(\lambda)$ at each~$\lambda \in \mathbb{R}$ (see the dotted curve in Figure~\ref{fig:TwoApprox}). By standard arguments, it follows that $\varphi$ is convex, piecewise linear, and continuous as it is the pointwise maximum of affine functions. For $m$~functions, the upper envelope can be computed in $\OO(m \log m)$~time as shown\footnotemark\xspace by \citet{HershbergerUpperEnvelope}. Within the same time bound, we can sort the resulting intervals by increasing values of their left boundary. Hence, we obtain a piecewise linear, continuous, and convex $\frac{1}{2}$-approximation $\varphi$ with $\OO(n^2)$~breakpoints in $\OO(n^2 \log n)$~time. In the following, we will refer to the intervals between the breakpoints of $\varphi$ as $I_1,\ldots,I_q$.
\footnotetext{Strictly speaking, the author proves the result for finite line segments and not for straight lines. However, it is easy to compute upper and lower bounds for the smallest and largest possible intersection point of two of the involved functions, respectively, and to reduce the problem to the resulting interval.}

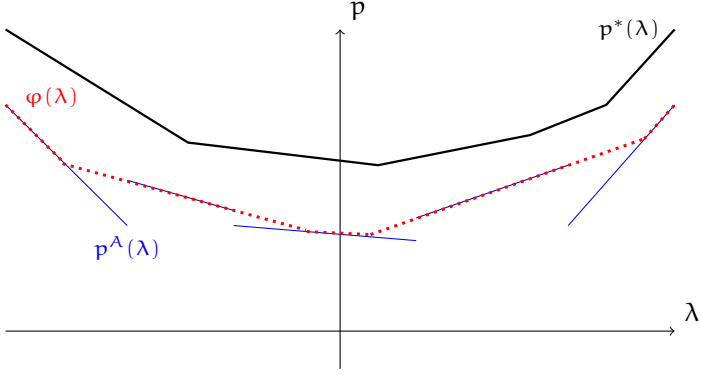
\begin{figure}[ht!]
	\begin{tikzpicture}
		\draw[->] (-4.4,0) -- (4.4,0) node[above right] {\small $\lambda$};
		\draw[->] (0,-0.5) -- (0,4) node[above right] {\small $p$};

		\draw[line width=0.3mm] (-4.4,4) -- (-2,2.5) -- (0.5,2.2) -- (2.5,2.6) -- (3.5,3.0) -- (4.4,4);
		\node at (3.8,4) {\scriptsize $p^*(\lambda)$};

		\draw[blue!100!white] (-4.4,3) -- (-2.8,1.4);
		\draw[blue!100!white] (-2.8,2.0) -- (-1.4,1.6);
		\draw[blue!100!white] (-1.4,1.4) -- (1.0,1.2);
		\draw[blue!100!white] (1.0,1.5) -- (3.0,2.2);
		\draw[blue!100!white] (3.0,1.4) -- (4.4,3.0);

		\draw[dotted,line width=0.4mm,red] (-4.4,3) -- (-3.65,2.25);
		% \draw[dotted,line width=0.4mm,red] (-3.65,2.25) -- (-2.8,2.0);
		% \draw[dotted,line width=0.4mm,red] (-1.4,1.6) -- (-0.45,1.35);
		\draw[dotted,line width=0.4mm,red] (-3.65,2.22) -- (-0.45,1.34);
		% \draw[dotted,line width=0.4mm,red] (0.4,1.28) -- (1.0,1.5);
		% \draw[dotted,line width=0.4mm,red] (3.0,2.2) -- (4.0,2.55);
		\draw[dotted,line width=0.4mm,red] (-0.44,1.32) -- (0.4,1.28);
		\draw[dotted,line width=0.4mm,red] (0.4,1.28) -- (4.0,2.55);
		\draw[dotted,line width=0.4mm,red] (4.0,2.55) -- (4.4,3.0);

		\node[red] at (-3.8,3.1) {\scriptsize $\varphi(\lambda)$};
		\node[blue] at (-2.8,1.1) {\scriptsize $p^A(\lambda)$};
	\end{tikzpicture}
	\caption{The optimal solution value~$p^*$ (black thick line), the profit function~$p^A$ of the $\frac{1}{2}$-approximation algorithm (blue straight lines), and the ``smoothed'' curve~$\varphi$ (red dotted lines).}
	\label{fig:TwoApprox}
\end{figure}

\section{Obtaining a parametric FPTAS}
\label{sec:FPTAS}

Before we explain the parametric FPTAS in detail, we first recapitulate the basic FPTAS for the traditional (non-parametric) knapsack problem as introduced by \citet{LawlerKnapsack} since its way of proceeding is crucial for the understanding of the parametric version.

\subsection{Traditional FPTAS}
\label{sec:FPTAS:Traditional}

Consider the case of some fixed value for $\lambda$, such that the profits have a constant, but possibly negative value~$p_i$. The basic FPTAS for the traditional knapsack problem is based on a well-known dynamic programming scheme, which was originally designed to solve the problem exactly in pseudo-polynomial time: Let $P$ denote an upper bound on the maximum profit of a solution to the given instance. For $k \in \{0,\ldots,n\}$ and $p \in \{0,\ldots,P\}$, let $w(k,p)$ denote the minimum weight that is necessary in order to obtain a profit of exactly $p$ with those items in the item set~$\{1,\ldots,k\}$ that have a non-negative\footnote{Note that items with a negative profit will not be present in an optimal solution.} profit. For $k=0$, we set $w(0,p) = 0$ for $p=0$ and $w(0,p) = W+1$ for $p > 0$. For $k \in \{1,\ldots,n\}$ and for the case that $p_k \in \{0,\ldots,p\}$, we compute the values~$w(k,p)$ recursively by $w(k,p) = \min\{w(k-1,p), w(k-1,p-p_k) + w_k\}$, representing the choice to either not pack the item or to pack it, respectively. Else, if $p_k \notin \{0,\ldots,p\}$, we set $w(k,p) = w(k-1,p)$ since we can omit negative item- and knapsack-profits. The largest value of $p$ such that $w(n,p) \leq W$ then yields the optimal solution to the problem. The procedure runs in pseudo-polynomial time $\mathcal{O}(nP)$.

The idea of the basic FPTAS is to scale down the item profits~$p_i$ to new values~$\pt_i \colonequals \floor{\frac{p_i}{M}}$, where $M \colonequals \frac{\varepsilon \cdot \o{p}}{n}$ for some value~$\o{p}$ fulfilling $\frac{1}{2} \cdot p^* \leq \o{p} \leq p^*$. Instead of setting $\o{p} \colonequals p^A$ as it is done in the traditional FPTAS, we can alternatively use our improved $2$-approximate solution~$\varphi$. Since the maximum possible profit~$\widetilde{P}$ is then given by
\begin{align*}
\widetilde{P} \leq \sum_{i=1}^n \pt_i \leq \sum_{i=1}^n \frac{n \cdot p_i}{\varepsilon \cdot \varphi} \leq \frac{n}{\varepsilon} \cdot \frac{p^*}{\varphi} \leq \frac{2n}{\varepsilon},
\end{align*}
the procedure runs in polynomial time~$\OO(\frac{n^2}{\varepsilon})$. The crucial observation is that we only lose a factor of $(1-\varepsilon)$ by scaling down the profits, so the solution obtained by the above dynamic programming scheme applied to the scaled profits yields a $(1-\varepsilon)$-approximate solution for the problem (see \citep{LawlerKnapsack,Knapsack} for further details on the algorithm).

\subsection{Parametric scaling}
\label{sec:FPTAS:Parametric}

Although the parametric FPTAS is based on the basic FPTAS, the instance parameters now depend on $\lambda$ and, thus, change while $\lambda$ increases. In particular, both the item profits and $\varphi$ now depend on $\lambda$, so the scaled profits~$\pt_i(\lambda) \colonequals \floor{\frac{n \cdot p_i(\lambda)}{\varepsilon \cdot \varphi(\lambda)}}$ have a highly non-linear behavior. Nevertheless, similar to the parametric $\frac{1}{2}$-approximation considered in Section~\ref{sec:TwoApprox:Parametric}, the solution returned by the FPTAS does not change as long as the profit~$\pt_i(\lambda)$ of each item remains constant. Hence, if $I'_j$ denotes an interval such that the scaled profits~$\pt_i$ remain constant for each $\lambda \in I'_j$, we can evaluate the dynamic programming scheme with the profits~$\pt_i$ to obtain a $(1-\varepsilon)$-approximate solution for the interval~$I'_j$. The proof of the polynomial running time and the approximation guarantee remain unchanged.

It remains to show how we can divide the real line into a polynomial number of intervals such that the profits remain constant in each interval. The basic FPTAS will then be evaluated for each such interval subsequently (using the corresponding constant scaled profits) in order to obtain $(1-\varepsilon)$-approximate solutions for the whole real line.

One natural idea would be to build on those intervals described in Section~\ref{sec:TwoApprox:Parametric} for which $\varphi(\lambda)$ behaves like an affine function: Let $I_1,\ldots,I_q$ with $q \in \mathcal{O}(n^2)$ denote the affine segments of $\varphi$ such that, for each $j \in \{1,\ldots,q\}$, the function $\varphi$ takes on some affine form~$\varphi(\lambda) = \alpha^{(j)} + \lambda \cdot \beta^{(j)}$ for $\lambda \in I_j$.

Now consider one such interval~$I_j$. If $\varphi(\lambda) = 0$ for $\lambda \in I_j$, it also holds that $p^*(\lambda) = 0$ since $\varphi(\lambda) \geq \frac{1}{2} \cdot p^*(\lambda)$ for $\lambda \in \mathbb{R}$, so the all-zero solution is optimal. Otherwise, for the non-rounded scaled profit of each item~$i$, it holds that
\[ \frac{n \cdot p_i(\lambda)}{\varepsilon \cdot \varphi(\lambda)} = \frac{n \cdot (a_i + \lambda \cdot b_i)}{\varepsilon \cdot (\alpha^{(j)} + \lambda \cdot \beta^{(j)})} \equalscolon \frac{n}{\varepsilon} \cdot f_i(\lambda). \]
These functions~$f_i$ are monotonous, since the first derivative fulfills
\begin{align*}
	\frac{df_i}{d\lambda} (\lambda) &= \frac{b_i \cdot (\alpha^{(j)} + \lambda \cdot \beta^{(j)}) - (a_i + \lambda \cdot b_i) \cdot \beta^{(j)}}{(\alpha^{(j)} + \lambda \cdot \beta^{(j)})^2} \\
									&= \frac{b_i \cdot \alpha^{(j)} - a_i \cdot \beta^{(j)}}{(\alpha^{(j)} + \lambda \cdot \beta^{(j)})^2}
\end{align*}
and, thus, does not change its sign within $I_j$. Hence, within the interval~$I_j$, each scaled profit~$\pt_i$ has a monotone behavior. Moreover, it holds that $0 \leq f_i(\lambda) \leq 2$ for all $\lambda \in I_j$ since each item either has a non-negative profit within the whole interval or it will be ignored and since $p_i(\lambda) \leq p^*(\lambda) \leq 2 \cdot \varphi(\lambda)$. These observations yield that each scaled profit~$\pt_i$ changes its (integral) value at most $\OO(\frac{n}{\varepsilon})$~times within $I_j$ since we only need to consider values for $\pt_i$ between $0$ and $\frac{2n}{\varepsilon}$. Hence, the above recursive formulae only change $\OO(\frac{n^2}{\varepsilon})$~times within each $I_j$, in which case we have to repeat the computation of the values~$w(i,p)$. This yields a total computational overhead of $\OO(\frac{n^4}{\varepsilon^2})$ per interval and, since there are at most $\mathcal{O}(n^2)$~intervals, a total running time of $\OO(\frac{n^6}{\varepsilon^2})$ for the parametric FPTAS. This running time will be significantly improved in the next subsection.

It should be noted that we need to take care of a proper definition of the returned intervals: For example, consider two scaled profits of the forms~$\pt_i = \floor{1 + \lambda}$ and $\pt_j = \floor{1 - \lambda}$. For the critical value~$\lambda_1 = 1$, both profits evaluate to $1$. However, for $\lambda_1' \colonequals \lambda_1 + \delta$ and $\lambda_1'' \colonequals \lambda_1 - \delta$ for a small value of $\delta$, one of the profits already changes its integral value and the dynamic programming scheme may behave differently. One simple solution is to assess that we add a single-point interval~$[\lambda_1,\lambda_1]$ for each critical value as well as two open intervals of the forms~$(\lambda_0,\lambda_1)$ and $(\lambda_1,\lambda_2)$, where $\lambda_0$ and $\lambda_2$ are adjacent critical values. The returned (ordered) sequence of intervals then alternates between single-point intervals and open intervals. For an open interval, we can obtain an approximate solution by setting $\lambda$ to the middle point of the interval and performing the dynamic program for the corresponding constant scaled profits.

% For each of the other intervals, we compute the critical values of $\lambda$ for which any of the scaled profits changes. Note that each profit can only take on values in $\{0,\ldots,\frac{n}{\varepsilon}\}$ since we can neglect items with negative profits on the one side and since $p^A(\lambda) \geq p_{max}(\lambda) \geq p_i(\lambda)$ due to its construction on the other side. For some $i \in \{1,\ldots,n\}$ and some value $k \in \{0,\ldots,\frac{n}{\varepsilon}\}$, it then follows that
% \begin{align*}
% 	\frac{n \cdot p_i(\lambda)}{\varepsilon \cdot p^A(\lambda)} = k &\Longleftrightarrow \frac{n \cdot (a_i + \lambda \cdot b_i)}{\varepsilon \cdot (a^{(j)} + \lambda \cdot b^{(j)})} = k \\
% 	&\Longleftrightarrow 
% \end{align*}
% Hence, for each of the $\OO(n^2)$ intervals, each item~$i \in \{1,\ldots,n\}$, and each possible profit in $\{0,\ldots,\frac{n}{\varepsilon}\}$, we can determine the critical values of $\lambda$ in 

\subsection{Improved Analysis}
\label{sec:FPTAS:Analysis}

The major drawback of the above algorithm is that we basically need to reset the whole procedure whenever the function~$\varphi$ changes its behavior. With this approach, we were able guarantee that each scaled profit has a monotone behavior such that each possible integral value is only attained at most once per interval. As it will be shown in this Section, we are somewhat allowed to ``ignore'' these changes without losing the guarantee that each possible value of the scaled profits will only be attained a constant number of times.

\begin{theorem}\label{thm:BoundedNumberOfValues}
	For each item~$i \in \{1,\ldots,n\}$, the scaled profits~$\pt_i(\lambda)$ attain each value in $\{ 0, \ldots, \frac{2n}{\varepsilon} \}$ at most three times as $\lambda$ increases from $-\infty$ to $+\infty$.
\end{theorem}

\begin{proof}
	In order to prove the claim, it suffices to show that the sign of the first derivative of each function~$f_i$ changes at most twice while $\lambda$ increases. As above, let $I_1,\ldots,I_q$ denote the intervals for which $\varphi$ takes on some affine form~$\varphi(\lambda) = \alpha^{(j)} + \lambda \cdot \beta^{(j)}$ such that
	\[ f_i(\lambda) = \frac{a_i + \lambda \cdot b_i}{\alpha^{(j)} + \lambda \cdot \beta^{(j)}} \]
	for $\lambda \in I_j$. Since $\varphi$ is convex and continuous, it holds that $\beta^{(j)} \leq \beta^{(j+1)}$ for $j \in \{1,\ldots,q-1\}$ and that there is some index~$h$ such that $\alpha^{(j)} \leq \alpha^{(j+1)}$ for $j \in \{1,\ldots,h\}$ and $\alpha^{(j)} \geq \alpha^{(j+1)}$ for $j \in \{h+1,\ldots,q-1\}$. In fact, due to the construction of the intervals, these inequalities hold in the strict sense since $\beta^{(j)} = \beta^{(j+1)}$ would also imply that $\alpha^{(j)} = \alpha^{(j+1)}$ by continuity of $\varphi$, so both segments would belong to the same interval. Hence, if we plot the points~$(\beta^{(j)},\alpha^{(j)})^T$ into a $b$-$a$-space, we get a picture as shown in Figure~\ref{fig:b_a_space}. Moreover, for each $j \in \{1,\ldots,q-1\}$, there is some $\lambda_j \in \mathbb{R}$ with
	\[ \alpha^{(j)} + \lambda_j \cdot \beta^{(j)} = \alpha^{(j+1)} + \lambda_j \cdot \beta^{(j+1)} \]
	due to the continuity and construction of $\varphi$. Hence, since $\beta^{(j+1)} = \beta^{(j)} + \delta_j$ for some value~$\delta_j > 0$, we get that
	\begin{align*}
		\alpha^{(j+1)} &= \alpha^{(j)} + \lambda_j \cdot \beta^{(j)} - \lambda_j \cdot \beta^{(j+1)} \\
		&= \alpha^{(j)} - \lambda_j \cdot \delta_j,
	\end{align*}
	so the slope of the line that connects the points~$(\beta^{(j)},\alpha^{(j)})^T$ and $(\beta^{(j+1)},\alpha^{(j+1)})^T$ evaluates to
	\begin{align*}
		\frac{\alpha^{(j+1)} - \alpha^{(j)}}{\beta^{(j+1)} - \beta^{(j)}} = \frac{-\lambda_j \cdot \delta_j}{\delta_j} = -\lambda_j
	\end{align*}
	and, thus, decreases while $j$ increases. This yields that the piecewise linear function~$g$ connecting each of the points~$(\beta^{(j)},\alpha^{(j)})^T$ in the order $j=1,\ldots,q$ is concave\footnotemark\xspace (as illustrated by the highlighted area in Figure~\ref{fig:b_a_space}).

	\footnotetext{This can also be seen by arguments used in the field of computational geometry: It is known that the upper envelope of a set of affine functions of the form~$c \cdot \lambda - d$ corresponds to the lower surface of a convex hull in the \emph{dual space}, which is clearly convex. Such a dual space contains a point~$(c,d)^T$ for each affine function of the above form in the primal space and, conversely, an affine function~$\lambda \cdot c - \mu$ for each point~$(\lambda,\mu)^T$ in the primal space. Hence, each line segment (breakpoint) of our upper envelope corresponds to a corner point (line segment) of the lower surface of a convex hull in the dual space (cf. \citep{BergComputationalGeometry}). In fact, Figure~\ref{fig:b_a_space} shows the dual space mirrored at the $b$-axis.}

	Now, for some specific item~$i \in \{1,\ldots,n\}$, consider the first derivate of $f_i$, which as we have seen evaluates to
	\[ \frac{df_i}{d\lambda}(\lambda) = \frac{b_i \cdot \alpha^{(j)} - a_i \cdot \beta^{(j)}}{(\alpha^{(j)} + \lambda \cdot \beta^{(j)})^2} \]
	as shown above. Since the denominator is always positive, we need to bound the number of times the sign of the numerator changes. The value $b_i \cdot \alpha^{(j)} - a_i \cdot \beta^{(j)}$ can be interpreted as the inner product of the vectors~$(-a_i,b_i)$ and $(\beta^{(j)},\alpha^{(j)})^T$. Hence, since $(-a_i,b_i) \cdot (b_i,a_i)^T = 0$, the sign of the derivate changes whenever the function $g$ crosses the line going through the origin and the point~$(b_i,a_i)^T$ (see the dotted line in Figure~\ref{fig:b_a_space}). Since $g$ is concave as shown above, this can happen at most two times while $\lambda$ increases, which yields the claim.
\qed\end{proof}

\begin{figure}[ht!]
	\begin{tikzpicture} \node[scale=0.90] {\begin{tikzpicture}
		\tikzset{cross/.style={cross out, draw=black, minimum size=2*(#1-\pgflinewidth), inner sep=0pt, outer sep=0pt}, cross/.default={2.5pt}}

		% Punkte-Dummies
		\foreach [count=\i] \x/\y in {-3/-2.5, -2.8/-1.5, -2.5/-0.5, -2.0/0.6, -1.5/1.2, -0.5/1.8, 0.5/1.8, 1.2/1.5, 1.8/1.0, 2.3/0.3, 2.7/-0.6, 3.0/-1.6} {
			\node (P\i) at (\x,\y) {};
		}

		% Fuellung
		\fill[black!10!white] (-3,-3) \foreach \i in {1,...,12}{ -- (P\i.center) } -- (3,-3) -- cycle;

		% Segmente
		\draw[black!50!white] \foreach \i [remember=\i as \lasti (initially 1)] in {2,...,12}{(P\lasti.center) -- (P\i.center)};

		% Punkte
		\foreach \i in {1,...,12}{
			\node[cross] at (P\i) {};
		}

		% Achsen
		\draw[->] (-4.4,0) -- (4.4,0) node[above right] {\small $b$};
		\draw[->] (0,-3) -- (0,3) node[above right] {\small $a$};

		% Beschriftung
		\node at (-2.0,-2.46) {\scriptsize $(\beta^{(j)},\alpha^{(j)})^T$};
		\node at (-1.5,-1.48) {\scriptsize $(\beta^{(j+1)},\alpha^{(j+1)})^T$};

		% Vektor und Hyperebene
		\draw[dotted,line width=0.3mm] (-4.4,2) -- (4.4,-2);
		\draw[->] (0,0) -- (-2.6,1.182);
		\node at (-2.2, 1.45) {\scriptsize $(b_i,a_i)^T$};
	\end{tikzpicture}}; \end{tikzpicture}
	\caption{Plotting the slope~$\beta^{(j)}$ and intersect~$\alpha^{(j)}$ of each affine segment of $\varphi$. The piecewise linear function connecting these points is concave and intersects with each straight line through the origin at most twice.}
	\label{fig:b_a_space}
\end{figure}
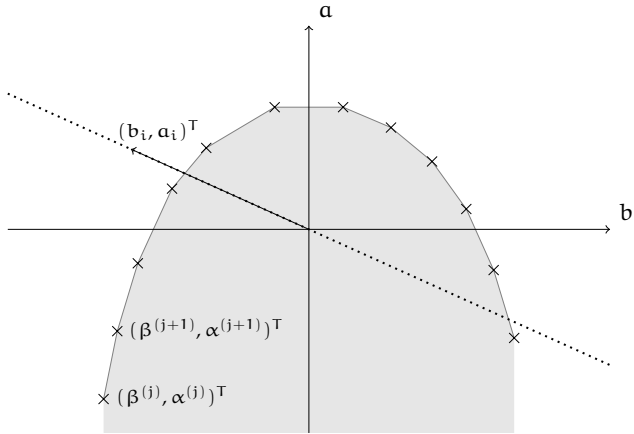

Theorem~\ref{thm:BoundedNumberOfValues} shows that each possible profit is only attained a constant number of times per item \emph{although} the involved functions~$f_i$ are rational functions whose denominator changes for increasing $\lambda$. Hence, each item only creates $\OO(\frac{n}{\varepsilon})$ subintervals as opposed to the $\OO(n^2 \cdot \frac{n}{\varepsilon})$~subintervals proven in Section~\ref{sec:FPTAS:Parametric}. It remains to show that we can determine these subintervals efficiently.

As shown in the proof of Theorem~\ref{thm:BoundedNumberOfValues}, the slope of each function~$f_i$ changes at most twice, yielding for each item up to three partitions of the set of intervals~$I_1,\ldots,I_q$ such that $f_i$ is monotonous within each partition. By scanning through the sequence of intervals of $\varphi$, we can determine these three partitions for all items in total time~$\OO(n \cdot n^2) = \OO(n^3)$. For each item, each partition, and each possible scaled profit in $\{0,\ldots,\frac{2n}{\varepsilon}\}$ (which can be attained only once in the partition), we perform a binary search on the intervals in the partition in order to find a value for $\lambda$ at which the corresponding profit is attained, if such a $\lambda$ exists. This can be done in $\OO(n \cdot 3 \cdot \frac{n}{\varepsilon} \cdot \log n^2) = \OO(\frac{n^2}{\varepsilon} \cdot \log n)$~time in total. Finally, we need to sort this list of critical values of $\lambda$ in order to determine the subintervals of the FPTAS, which can be done in $\OO(n \cdot \frac{n}{\varepsilon} \cdot \log (n \cdot \frac{n}{\varepsilon})) = \OO(\frac{n^2}{\varepsilon} \cdot \log \frac{n}{\varepsilon})$~time.

In summary, we need $\OO(n^3 + \frac{n^2}{\varepsilon} \cdot \log \frac{n}{\varepsilon})$~time to determine the $\OO(\frac{n^2}{\varepsilon})$~subintervals of the FPTAS. For each of these subintervals, we need to perform the traditional FPTAS with the corresponding (constant) scaled profits, which can be done in $\OO(\frac{n^2}{\varepsilon})$~time each. Hence, we obtain an FPTAS for the parametric knapsack problem running in $\OO(\frac{n^4}{\varepsilon^2})$~time in total. This yields the main result of this paper:

\begin{theorem}
	There is an FPTAS for the parametric knapsack problem running in strongly polynomial time $\OO(\frac{n^4}{\varepsilon^2})$. \qed
\end{theorem}

\section{Combining FPTASs}
\label{sec:FPTAS2}

In the previous section, we have seen that the parametric knapsack problem can be divided into $\OO(\frac{n^2}{\varepsilon})$~subproblems, for which we need to provide $(1 - \varepsilon)$-approximate solutions. These subproblems were created in a way such that the scaled profits are constant for each subproblem. Each of them can be seen as a new, independent, and non-parametric knapsack instance (albeit a special one, since the profits are now of polynomial size). In Section~\ref{sec:FPTAS}, we simply solved each of the resulting knapsack instances exactly in $\mathcal{O}(\frac{n^2}{\varepsilon})$~time. The main observation of this section is that we actually do not necessarily need to solve the subproblems exactly -- it suffices to solve them up to a factor of $(1 - \varepsilon)$ using \emph{any} FPTAS for the traditional knapsack problem.

Consider one fixed interval~$I'$ of the $\OO(\frac{n^2}{\varepsilon})$~subintervals of the problem. For each $\lambda \in I'$, the scaled profits~$\pt_i$ take on constant values. Moreover, it holds that $\frac{1}{2} \cdot p^*(\lambda) \leq \varphi(\lambda) \leq p^*(\lambda)$ for any $\lambda \in I'$ as shown in Section~\ref{sec:TwoApprox:Parametric}. Let $\o{x}$ denote a solution returned by some FPTAS for the traditional knapsack problem that is called on an instance with the scaled profits and let $x$ denote an exact solution to the scaled instance (which, e.g., can be obtained by the dynamic programming scheme as above). Clearly, it holds that
\[ \widetilde{p}(\o{x}) \colonequals \sum_{i=1}^n \pt_i \cdot \o{x}_i \geq (1 - \varepsilon) \cdot \sum_{i=1}^n \pt_i \cdot x_i \equalscolon \widetilde{p}(x). \]
For any fixed $\lambda \in I'$ and an optimal solution~$x^*$ for the unscaled problem at $\lambda$, we then get the following approximation guarantee for the solution~$\o{x}$:
{\allowdisplaybreaks
\begin{align*}
	p(\o{x}) &= \sum_{i=1}^n p_i \cdot \o{x}_i \geq \sum_{i=1}^n M \cdot \floor{\frac{p_i}{M}} \cdot \o{x}_i \\
	&= M \cdot \widetilde{p}(\o{x}) \\
	&\geq (1 - \varepsilon) \cdot M \cdot \widetilde{p}(x) \\
	&\geq (1 - \varepsilon) \cdot M \cdot \widetilde{p}(x^*) \\
	&\geq (1 - \varepsilon) \cdot M \cdot \sum_{i=1}^n \left(\frac{p_i}{M} - 1 \right) \cdot x^*_i \\
	&= (1 - \varepsilon) \cdot \left(p^*(\lambda) - M \cdot \sum_{i=1}^n x^*_i \right) \\
	&\geq (1 - \varepsilon) \cdot \left(p^*(\lambda) - \varepsilon \cdot \varphi(\lambda) \right) \\
	&\geq (1 - \varepsilon) \cdot \left(p^*(\lambda) - \varepsilon \cdot p^*(\lambda) \right) \\
	&\geq (1 - \varepsilon)^2 \cdot p^*(\lambda) \\
	&= (1 - 2\varepsilon + \varepsilon^2) \cdot p^*(\lambda) \\
	&\geq (1 - 2\varepsilon) \cdot p^*(\lambda).
\end{align*}%
}%
Setting $\varepsilon' \colonequals \frac{\varepsilon}{2}$ then yields the desired approximation guarantee. Hence, although the subproblems were designed in a way such that the basic dynamic programming scheme does not change its behavior, we do not necessarily need to execute it but can also use an FPTAS instead.

\begin{theorem}
	There is an FPTAS for the parametric knapsack problem running in $\OO(\frac{n^2}{\varepsilon} \cdot A(n,\varepsilon))$~time, where $A(n,\varepsilon)$ denotes the running time of an FPTAS for the traditional knapsack problem. \qed
\end{theorem}
Note that it clearly holds that $A(n,\varepsilon) \in \Omega(n)$, so the running time of the main procedure will dominate the overheads to compute~$\varphi$ and the set of subintervals.

At present, the best FPTAS for the traditional knapsack problem is given by \citet{KellererPferschyFPTAS,KellererPferschyFPTAS2} and achieves a running time of
\begin{align*}
	\OO\left( n \cdot \min\left\{\log n, \log \frac{1}{\varepsilon} \right\} + \right. \quad\quad\\
	\left. \frac{1}{\varepsilon^2} \log \frac{1}{\varepsilon} \cdot \min\left\{ n, \frac{1}{\varepsilon} \log \frac{1}{\varepsilon} \right\} \right).
\end{align*}
Under the commonly used assumption that $n$ is much larger than $\frac{1}{\varepsilon}$ in practice \citep{LawlerKnapsack}, this running time evaluates to
\begin{align*}
	\OO\left( n \log \frac{1}{\varepsilon} + \frac{1}{\varepsilon^3} \log^2 \frac{1}{\varepsilon} \right),
\end{align*}
yielding an FPTAS for the parametric knapsack problem with a strongly polynomial running time of
\begin{align*}
	\OO\left( \frac{n^3}{\varepsilon} \log \frac{1}{\varepsilon} + \frac{n^2}{\varepsilon^4} \log^2 \frac{1}{\varepsilon} \right).
\end{align*}

\bibliographystyle{elsarticle-num-names}
\bibliography{/Users/holzhaus/Documents/Forschung/literature}

\end{document}